\documentclass[journal,transmag]{IEEEtran}
\usepackage{setspace}

\usepackage{ifpdf}

\ifCLASSINFOpdf
   \usepackage[pdftex]{graphicx}
\else
   \usepackage[dvips]{graphicx}
\fi
\usepackage[cmex10]{amsmath}
\usepackage{amssymb,amsmath}

\newtheorem{theorem}{Theorem}
\newtheorem{lemma}{Lemma}

\newtheorem{problem}{Problem}

\newtheorem{proposition}{Proposition}
\newtheorem{definition}{Definition}
\newtheorem{corollary}{Corollary}

\newtheorem{example}{Example}

\begin{document}
%
\title{When can dictionary learning uniquely \\ recover sparse data from subsamples?}
%
%
%
\author{Christopher J. Hillar,
        Friedrich T. Sommer
\thanks{C. J. Hillar conducted the research while at the Mathematical Sciences Research Institute (MSRI), Berkeley, CA, USA and the Redwood Center for Theoretical Neuroscience, UC Berkeley, CA, e-mail: chillar@berkeley.edu.  F. T. Sommer is with the Redwood Center for Theoretical Neuroscience at UC Berkeley, CA, e-mail: fsommer@berkeley.edu.
}}

\maketitle

\begin{abstract}
Sparse coding or sparse dictionary learning has been widely used to recover underlying structure in many kinds of natural data.  Here, we provide conditions guaranteeing when this recovery is universal; that is, when sparse codes and dictionaries are unique (up to natural symmetries).  Our main tool is a useful lemma in combinatorial matrix theory that allows us to derive bounds on the sample sizes guaranteeing such uniqueness under various assumptions for how training data are generated. Whenever the conditions to one of our theorems are met, any sparsity-constrained learning algorithm that succeeds in reconstructing the data recovers the original sparse codes and dictionary. We also discuss potential applications to neuroscience and data analysis.
\end{abstract}




%

\section{Introduction}
%
%
%
%

\IEEEPARstart{I}{ndependent} component analysis \cite{comon1994independent, bell1995information} and dictionary learning with a sparse coding scheme \cite{olshausenfield1996} have become important tools for revealing underlying structure in many different types of data. The common goal of these algorithms is to produce a latent representation of data that exposes  underlying structure.  Such a map from data to representations is often called \textit{coding} or \textit{inference} and the reverse map from representations to estimated data is called \textit{reconstruction}. In the coding step of \cite{olshausenfield1996}, for instance, given a data point $\mathbf{y}$ in a dataset $Y$, one computes a \textit{sparse} code vector $\mathbf{b} = f(\mathbf{y})$ with a small number of nonzero coordinates.  This code vector $\mathbf{b}$ is used to linearly reconstruct $\mathbf{y}$ as 
\begin{equation}\label{line_gen_model1}
\hat{\mathbf{y}} = B  \mathbf{b},
\end{equation}
using a reconstruction matrix $B$ (often called a \textit{basis} or \textit{dictionary} for $Y$). 

The code map $f(\mathbf{y})$ and the matrix $B$ are fit to training data using unsupervised learning.
If reconstruction succeeds and $\hat{\mathbf{y}} = \mathbf{y}$ for all data in $Y$, then representations $\mathbf{b}$ and the dictionary $B$ capture essential structure in the data.  Otherwise, $(\hat{\mathbf{y}} - \mathbf{y})$ provides an error signal to improve the two steps of coding and reconstruction.  This control loop for optimizing a data representation is often referred to as ``self-supervised learning" or ``auto-encoding"  (e.g., \cite{hinton1989}).  In the context of theoretical neuroscience, the sparse vector $\mathbf{b}$ is also sometimes called an ``efficient representation'' of $\mathbf{y}$ because it exposes the coefficients of the ``independent components", thereby minimizing redundancy of description \cite{Attneave1954, barlow1959}.

In the literature, self-supervised learning with a sparseness constraint is called \textit{sparse coding} or \textit{sparse dictionary learning} (SDL).
It has been found empirically that SDL succeeds at revealing (unique) structure for a wide range of sensory data, most notably natural images \cite{olshausenfield1996, hurri1996image, bell1997independent, van1998independent, rehnsommer2007}, natural sounds \cite{bellsejnowski1996, smithlewicki2006, Carlson12}, and even artistic output \cite{hughes2010}.   Importantly for applications, SDL can also reveal \textit{overcomplete} representations; that is, the dimension of $\mathbf{b}$ can exceed that of the data. For example, SDL can reveal structure in data that has been sparsely composed from multiple dictionaries \cite{chen2001}.   

Assume that observed data $Y$ were generated as
\begin{equation}\label{generated_data}
\mathbf{y} = A  \mathbf{a},
\end{equation}
using a fixed \mbox{$n \times m$} \textit{generation matrix} $A$ and sparse $m$-dimensional latent representations $\mathbf{a}$. 
Here we investigate conditions under which the generation matrix and the sparse representations can be recovered from samples $Y$. Specifically, we ask: 
How much data is required for self-supervised learning to uniquely discover $A$ and the corresponding latent representations $\mathbf{a}$? How overcomplete can the latent representation be and still be uniquely identified? 

This study extends earlier pioneering work on uniqueness of sparse matrix factorizations in dictionary learning \cite{aharon2006} although our initial motivation was  to explain findings in theoretical \cite{coulter2010, NIPS10} and computational neuroscience \cite{agarwal2014spatially}.  Specifically, it was proposed in \cite{coulter2010} and \cite{NIPS10} that SDL could be used as a model for how neurons in sensory areas communicate high-dimensional sparse representations through wiring bottlenecks.  
Assume a rather small number of projecting neurons subsample the sensory representations in one brain region and send the signals downstream to another brain region.  In this case, uniqueness of SDL guarantees that self-supervised learning in the downstream region automatically recovers the full original sensory representations. 

Before addressing our main questions, it is informative to first explain their relationship to \textit{compressed sensing} (CS), a recent advance in signal processing \cite{donoho2006compressed, candes2005decoding}.
\begin{figure}
\begin{center}
\textbf{a}) \includegraphics[height=.58 \linewidth]{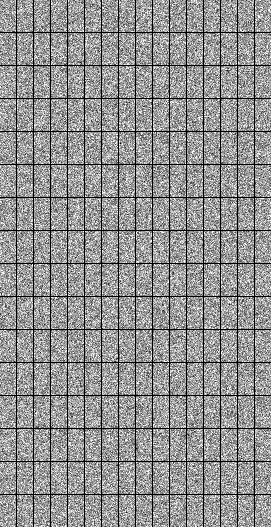}
\textbf{b}) \includegraphics[height=.58 \linewidth]{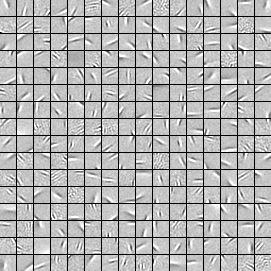}
\end{center}
\caption{\textbf{Sparse coding in a compressed space} reveals original sparse representations but not original dictionaries \cite{coulter2010, NIPS10}. \textbf{a)} Dictionary columns trained in $2 \times$ compressed space $\mathbf{y} = \Phi \mathbf{x}$ starting from $16 \times 16$ natural image patches $\mathbf{x}$ appear random (compression matrix $\Phi$ was formed from i.i.d. standard normals). \textbf{b)} Cross-correlating $16 \times 16$ natural image patches $\mathbf{x}$ with the coordinates of inferred sparse vectors $\mathbf{b} = f( \mathbf{y})$ after compression recovers original sparse coding ``Gabor" dictionary $\Psi$ for natural patches \cite{olshausenfield1996}.}
\label{ACS_fig_random}
\end{figure}
The theory of CS provides techniques to recover data vectors $\mathbf{x}$ with sparse structure after they have been \textit{subsampled} as $\mathbf{y} = \Phi \mathbf{x}$ by a known \textit{compression} matrix $\Phi$.  The sparsity usually enforced is that $\mathbf{x}$ can be expressed as $\mathbf{x} = \Psi \mathbf{a}$ using a known dictionary matrix $\Psi$ and an $m$-dimensional vector $\mathbf{a}$ with at most $k \ll m$ nonzero  entries.  Such vectors $\mathbf{a}$ are called \textit{$k$-sparse}.

Under mild CS conditions on the generation matrix:
\begin{equation}
A = \Phi \Psi,
\label{A_mat}
\end{equation}
which in this case involves both the compression matrix and the sparse dictionary for the data, the theory gives accurate recovery of $\mathbf{a}$ in equation (\ref{generated_data}) (and thus $\mathbf{x}$) from the compressed vector $\mathbf{y}$
 as long as the dimension $n$ of $\mathbf{y}$ satisfies:
\begin{equation}\label{recovcond}
n \geq C k \log (m/k).
\end{equation}
Here, $C$ is a constant independent of $m$, $n$, and $k$.\footnote{For a more detailed discussion of these facts (including proofs) and their relationship to approximation theory, we refer to \cite{Baraniuk2008} and its references.}
In other words, one can easily recover sparse high-dimensional representations $\mathbf{a}$  from ``projections" (\ref{generated_data}) into spaces with a dimension (almost) as small as the count of $\mathbf{a}$'s active entries.

A typical CS assumption on the linear transformation $A \in \mathbb R^{n \times m}$ is the \textit{spark condition}:
\begin{equation}\label{CSassump}
A \mathbf{a}_1 = A \mathbf{a}_2  \ \ \text{for $k$-sparse}  \ \mathbf{a}_1,\mathbf{a}_2 \in \mathbb R^m \ \   \Longrightarrow \  \  \mathbf{a}_1 =  \mathbf{a}_2.
\end{equation}
Note that a generic square matrix $A$ is invertible; thus, (\ref{CSassump}) trivially holds for almost all matrices $A$ whenever $n = m$.  In the interesting regimes of compressed sensing, however, the sample dimension $n$ is significantly smaller than the original data dimension $m$.  Thus, condition (\ref{CSassump}) supplants invertibility of the matrix $A$ with an ``incoherence" among every $2k$ of its columns.  Rather remarkably, even in the critical regime close to equality of (\ref{recovcond}), condition (\ref{CSassump}) holds with very high probability for many ensembles of randomly generated $A$.  In particular, it is easy to produce matrices  satisfying (\ref{CSassump}).

Conditions (\ref{recovcond}) and (\ref{CSassump}) are cornerstones of CS theory, describing when it is possible to recover a code vector $\mathbf{a}$ from a  subsampled measurement $\mathbf{y}$ generated by (\ref{generated_data}) using a known generation matrix $A$.  The dictionary learning problem is  even more difficult: to recover code vectors {\it and} the generation matrix from measurements.  As we shall see, however, the same spark assumption for $A$ in CS theory also guarantees universality of SDL, given enough training samples.

Taken literally, of course, uniqueness in SDL is ill-posed \cite{gleichman2011}.  More precisely, if $P$ is a permutation matrix\footnote{A \textit{permutation matrix} $P$ has binary entries and precisely one $1$ in each row and column (thus, $P\mathbf{v}$ for a column vector $\mathbf{v}$ permutes its entries). Note that $PP^{\top} = P^{\top}P = I$, where $I$ denotes the $m \times m$ identity matrix, and $M^{\top}$ for a matrix $M$ is its transpose.  In particular, we have $P^{-1} = P^{\top}$.} and $D$ is an invertible diagonal matrix, then \[A\mathbf{a} = (AD^{-1}P^{\top})(PD\mathbf{a})\] for each sample $\mathbf{y} = A \mathbf{a}$.  Thus, without access to $A$, one could not discriminate which of $\mathbf{a}$ or $PD\mathbf{a}$ (resp. $A$ or $AD^{-1}P^{\top}$) was the original sparse vector (resp. generation matrix).   This discussion motivates the following definition and problem.

\begin{definition}
When a dataset $Y = \{A \mathbf{a}_1,\ldots,A\mathbf{a}_N\}$ with $k$-sparse $\mathbf{a}_i$ has the property that any other $k$-sparse coding of it is the same up to a uniform permutation and invertible scaling, we say that $Y$ has a \textit{unique sparse coding}.
\end{definition}

\begin{problem}\label{main_pr}
Let $Y = \{\mathbf{y}_1, \dots , \mathbf{y}_N\}$ be generated by $N$ subsamplings $\mathbf{y}_i = A\mathbf{a}_i$ as in (\ref{generated_data}), where \mbox{$A \in \mathbb R^{n \times m}$} satisfies the CS spark condition (\ref{CSassump}) and the $\mathbf{a}_i$ are $k$-sparse.  When does $Y$ have a unique sparse coding?
\end{problem}

In equations, Problem \ref{main_pr} asks when we can guarantee that if  \begin{equation}\label{AaeqBbi}
A \mathbf{a}_i = \mathbf{y}_i =  \hat{\mathbf{y}}_i = B \mathbf{b}_i, \ \ \ i = 1,\ldots,N,
\end{equation}
is another coding of $Y$ using some other matrix $B$ (not necessarily satisfying the spark condition) with $k$-sparse codes $\mathbf{b}_i$, then there is a permutation matrix $P$ and an invertible diagonal matrix $D$ with
\begin{equation}\label{PDaeqB}
A = BPD \ \ \ \text{and} \ \ \mathbf{b}_i = PD \mathbf{a}_i,  \ \  \ i = 1,\ldots,N.
\end{equation}

Here, we determine bounds on the number of samples $N$ guaranteeing that datasets $Y$ have unique sparse codings under various assumptions for how sparse $\mathbf{a}_i$ are produced;  Table \ref{table_N} gives a summary. 
Our main result is the following.

\begin{table*}[!t]
\begin{center}
$\begin{array}{c|cc} \text{Samples } Y = \{A\mathbf{a}_1, \ldots,A\mathbf{a}_N\} \text{ with $k$-sparse } \mathbf{a}_i \in \mathbb R^m \text{ chosen as follows:}   & \text{Sufficient samples } N \text{ for unique recovery:}  \\\hline 
\text{$k{m \choose k}$ in general position from each support set, independent of $A$ satisfying (\ref{CSassump})}  & k{m \choose k}^2 \\\hline     
 \text{given the matrix $A$ satisfying (\ref{CSassump}), draw $k+1$ randomly from each support set}  & (k+1){m \choose k} \ \ \text{(with probability one)} \\\hline 
\text{given the matrix $A$ satisfying (\ref{CSassump}), draw randomly uniformly over supports}  & \frac{k+1}{\beta}{m \choose k} \ \ \text{(with probability $1-\beta$)} \\\hline 
\text{randomly draw $k+1$ from each support set, independent of $A \notin Z$ 
}  & (k+1){m \choose k} \ \ \text{(with probability one)} \\\hline 
\text{randomly draw uniformly over supports, independent of $A \notin Z$ 
}  &  \frac{k+1}{\beta}{m \choose k} \ \ \text{(with probability $1-\beta$)} \\\hline 
\end{array}$
\end{center}
\caption{\textrm{\normalfont 
There are ${m \choose k}$ different support sets specifying which $k$ entries of $\mathbf{a} \in \mathbb R^m$ are nonzero.  Given a support set, a ``randomly drawn" sparse vector $\mathbf{a}$ has $k$ i.i.d. uniform $[0,1]$ random variables placed as the nonzero components determined by the support set.  The set $Z \subset \mathbb R^{n \times m}$ has Lebesgue measure zero. }}
\label{table_N}
\end{table*}

\begin{theorem}\label{unique_thm1}
There exist $N = k{m \choose k}^2$ $k$-sparse vectors $\mathbf{a}_1,\ldots,\mathbf{a}_{N} \in \mathbb R^m$ such that any matrix $A \in \mathbb R^{n \times m}$ satisfying spark condition (\ref{CSassump}) gives rise to a dataset $Y = \{A\mathbf{a}_1, \ldots, A\mathbf{a}_N\}$ having a unique sparse coding.  
\end{theorem}

In fact, there are many such sets of (deterministically produced) $\mathbf{a}_i$;  we give a parameterized family in Section \ref{acsTheoremSection}.

Our next result exploits randomness to give a different construction requiring fewer samples.  Fix a matrix $A$ satisfying CS condition (\ref{CSassump}) and a small $\beta > 0$.  There is a simple random drawing procedure (Definition \ref{random_def} below) for generating $N = \frac{k+1}{\beta}{m \choose k}$ $k$-sparse vectors $\mathbf{a}_1,\ldots,\mathbf{a}_{N}$ such that  $Y = \{A\mathbf{a}_1, \ldots, A\mathbf{a}_N\}$ has a unique sparse coding (with probability at least $1 - \beta$).  
In particular, sparse matrix factorization of sparsely coded data is typically unique for a large enough number of random samples.

A subtle difference between these two results is that the $\mathbf{a}_1,\ldots,\mathbf{a}_{N}$ in Theorem \ref{unique_thm1} are independent of the generation matrix $A$.  The following statement is the best we are able to do in this direction using random methods.   With probability one, a random draw of $(k+1)$ $k$-sparse $\mathbf{a}_i$ from each of the ${m \choose k}$ support sets satisfies the following property.  There is a set of matrices $Z \subset \mathbb R^{n \times m}$ with Lebesgue measure zero such that if $A \notin Z$, then $Y = \{A\mathbf{a}_1,\ldots,A\mathbf{a}_{N}\}$ has a unique sparse coding.  In other words, almost surely such a set of $\mathbf{a}_i$ will have almost every $A$ give rise to $Y$ with unique sparse codings.

The most general previous uniqueness result on this problem is given in \cite[Theorem 3]{aharon2006} with an argument using the singular value decomposition for matrices.  Assuming additionally that the matrix $B$ also obeys CS condition (\ref{CSassump}) and that $\mathbf{b}_1,\ldots,\mathbf{b}_{N}$ satisfy certain assumptions\footnote{For completeness, we list them here:  (i) the \textit{supports} (the indices of nonzero components) of each $\mathbf{a}_i$ and $\mathbf{b}_i$ consist of exactly $k$ elements, (ii) for each possible $k$-element support, there are at least $k+1$ vectors $\mathbf{a}_i$ (resp. $\mathbf{b}_i$) having this support, (iii) any $k + 1$ vectors $\mathbf{a}_i$ (resp. $\mathbf{b}_i$) having the same support span a $k$-dimensional space, and (iv) any $k + 1$ vectors $\mathbf{a}_i$ (resp. $\mathbf{b}_i$) having different supports span a $(k + 1)$-dimensional space.}\label{aharon_cond}, the authors of \cite{aharon2006} show that $(k+1){m \choose k}$ many equations (\ref{AaeqBbi}) imply (\ref{PDaeqB}). 
With Theorem \ref{unique_thm1}, we close a gap in the literature by finding the weakest possible assumptions on recovery codes and dictionaries that guarantee uniqueness (although there is room for improvement in the size $N$ of our set of samples).

As a motivating example of Theorem \ref{unique_thm1}, consider applying SDL on grayscale natural image patches $\mathbf{x}$. It is known empirically that after SDL converges, any image patch can be represented as $\mathbf{x} = \Psi \mathbf{a}$ where $\mathbf{a}$ is a sparse vector and the matrix $\Psi$ contains two-dimensional Gabor functions \cite{olshausenfield1996} (as displayed in Fig.~\ref{ACS_fig_random}\textbf{b}). Consider now SDL on subsampled natural image patches $\mathbf{y} = \Phi \mathbf{x} = \Phi \Psi \mathbf{a}$, with $\Phi$ a compression matrix. If $A = \Phi \Psi$ satisfies the spark condition, then any successful sparse coding $\hat{\mathbf{y}} = B \mathbf{b}$ of enough such $\mathbf{y}$ should produce codes $\mathbf{b} = f(\mathbf{y})$ that are equal (up to symmetry) to the sparse vectors $\mathbf{a}$ that resulted from learning on the full image patches.  Indeed, this is the empirical finding \cite{NIPS10}, even when we choose $\mathbf{x}$ uniformly from a database of natural image patches \cite{van1998independent} rather than generating them from special $\mathbf{a}$. Note that although the dictionary trained on compressed images appears uninformative (Fig.~\ref{ACS_fig_random}\textbf{a}), the cross-correlation of the inferred sparse vectors $\mathbf{b}$ with the full image patches contains two-dimensional Gabor functions (Fig.~\ref{ACS_fig_random}\textbf{b}). In fact, the cross-correlation matrix is equal to the dictionary $\Psi$ for constructing the original patches (again, up to natural symmetries).

The organization of this paper is as follows.  In Section \ref{acsTheoremSection}, we state our main tool from combinatorial matrix theory (Lemma \ref{subspaceLem_}) and then derive from it Theorem \ref{unique_thm1}.
We next provide precise statements in Section  \ref{proof_prob_Section} of our randomized theorems and then prove them in Section \ref{prob_proofs}.  The final section gives a short discussion of how our results fit into the general sparse dictionary learning literature and how they might apply to neuroscience.  An appendix contains the proof of Lemma \ref{subspaceLem_}.

\section{Deterministic uniqueness theorem}\label{acsTheoremSection}

In what follows, we will use the notation $[m]$ for the set $\{1,\ldots,m\}$, and ${ [m] \choose k}$ for the set of $k$-element subsets of $[m]$.  Also, recall that Span$\{\mathbf{v}_1,\ldots,\mathbf{v}_{\ell}\}$ for real vectors $\mathbf{v}_1,\ldots,\mathbf{v}_{\ell}$ is the vector space consisting of their $\mathbb R$-linear span: \[ \text{Span}\{\mathbf{v}_1,\ldots,\mathbf{v}_{\ell}\} = \left\{ \sum_{i=1}^{\ell} t_i \mathbf{v}_i : \ t_1,\ldots,t_{\ell} \in \mathbb R \right\}.\]  Finally, for a subset $S \subseteq [m]$ and a matrix $A$ with columns $\{A_1,\ldots,A_m\}$, we define \[ \text{\rm Span}\{A_S\} = \text{\rm Span}\{A_{s}: s \in S\}. \]

Before proving Theorem \ref{unique_thm1} in full generality, it is illustrative to consider when $k = 1$.  In this simple case, we only need $N = m$ samples for uniqueness.  Set $\mathbf{a}_i = \mathbf{e}_i$ ($i= 1,\ldots,m$) to be the standard basis column vectors in $\mathbb R^m$.  Assuming that (\ref{AaeqBbi}) holds for some matrix $B$ and $1$-sparse $\mathbf{b}_i$, it follows that
\begin{equation}\label{keqonecasef}
A\mathbf{e}_{i} =  B\,c_i \mathbf{e}_{\pi(i)}, \ \ i = 1,\ldots,m,
\end{equation}
for some map $\pi: \{1,\ldots,m\} \to \{1,\ldots,m\}$ and $c_i \in \mathbb R$. Note that if $c_i = 0$ for some $i$, then $A\mathbf{e}_i = 0$, contradicting (\ref{CSassump}).  

Next, we show that $\pi$ is necessarily injective (and thus is a permutation).
Suppose that $\pi(i) = \pi(j)$; then, \[A\mathbf{e}_i = c_iB\mathbf{e}_{\pi(i)} = c_i B\mathbf{e}_{\pi(j)} = \frac{c_i}{c_j} Bc_j \mathbf{e}_{\pi(j)} = \frac{c_i}{c_j} A \mathbf{e}_j.\]
Again by (\ref{CSassump}) this is only possible if $i = j$.  Thus, $\pi$ is a permutation.

Let $P$ and $D$ denote the following permutation and diagonal matrices, respectively:
\begin{equation}\label{permeqn}
P = \left( \mathbf{e}_{\pi(1)} \cdots \mathbf{e}_{\pi(m)}\right), \ \ D = \left(\begin{array}{ccc}c_1 & \cdots & 0 \\\vdots & \ddots & \vdots \\0 & \cdots & c_m\end{array}\right).
\end{equation}
The matrix formed by stacking left-to-right the column vectors on the right-hand side of (\ref{keqonecasef}) is easily seen to satisfy: \[ \left(c_1 B\mathbf{e}_{\pi(1)} \cdots c_m B \mathbf{e}_{\pi(m)}\right) = BPD.\]  On the other hand, the columns $A\mathbf{e}_i$ form the matrix $A$.
Taken together, therefore, equations (\ref{keqonecasef}) imply that $A = BPD$.  

Note that the identity $A = BPD$ already implies in general the recovery result $\mathbf{b}_i = PD\mathbf{a}_i$ for all $i$.  This follows since \[A \mathbf{a}_i = B\mathbf{b}_i = A (D^{-1}P^{\top}\mathbf{b}_i)\] gives $\mathbf{b}_i = PD \mathbf{a}_i$ from (\ref{CSassump}) because both $\mathbf{a}_i$ and  $D^{-1}P^{\top}\mathbf{b}_i$ are $k$-sparse vectors.  

Unfortunately, the proof for larger $k$ is more challenging.  The difficulty is that in general it is nontrivial to produce $P$ and $D$ as in (\ref{permeqn}) using only assumptions (\ref{CSassump}) and (\ref{AaeqBbi}).   For this, our main results depend on combinatorial linear algebra.

\begin{lemma}[Main Lemma]\label{subspaceLem_}
Fix positive integers $n$ and $k < m$.  Let $A \in \mathbb R^{n \times m}$ and $B \in \mathbb R^{n \times m}$ have columns $\{A_1,\ldots,A_m\}$ and $\{B_1,\ldots,B_m\}$, respectively.  Suppose that $A$ satisfies spark condition (\ref{CSassump}) and that $\pi: {[m] \choose k} \to {[m] \choose k}$ 
is a map such that 
\begin{equation}\label{eq:spanhypoth}
\text{\rm Span}\{A_S\} = \text{\rm Span}\{B_{\pi(S)}\}, \ \ \text{for all } S \in {[m] \choose k}.
\end{equation}
Then there exists a permutation matrix $P \in \mathbb R^{m \times m}$ and an invertible diagonal matrix $D \in \mathbb R^{m \times m}$ such that $A = BPD$.
In particular, the matrix $B$ also satisfies the spark condition.
\end{lemma}

We defer the proof of this lemma for $k > 1$ to the Appendix.

\begin{IEEEproof}[Proof of Theorem \ref{unique_thm1}]
First, we produce a set of $N = k {m \choose k}^2$ vectors $\mathbf{s}_i \in \mathbb R^k$ in general linear position (i.e., any subset of $k$ of them are linearly independent).  Specifically, let $\alpha_1, \ldots, \alpha_{N}$ be any distinct numbers.  Then the columns of the $k \times N$ matrix $V = \left(\alpha_j^i\right)_{i,j=1}^{k, N}$ are in general linear position (since the $\alpha_j$ are distinct, any $k \times k$ ``Vandermonde" sub-determinant is nonzero). 
Next, form the $k$-sparse vectors $\mathbf{a}_1,\ldots,\mathbf{a}_N \in \mathbb R^m$ by taking $\mathbf{s}_i$ for the support values of $\mathbf{a}_i$ where each possible support set is represented $k{m \choose k}$ times.

We claim that these $\mathbf{a}_i$ always produce unique sparse codings. To prove this, suppose that $A$ satisfies the spark condition and set $\mathbf{y}_i = A\mathbf{a}_i$.  By general linear position and the spark condition, for every subset of $k$ vectors $\mathbf{y}_{i_1},\ldots,\mathbf{y}_{i_k}$ generated using the same support $S$, we have Span$\{ \mathbf{y}_{i_1},\ldots,\mathbf{y}_{i_k}\} = \text{Span}\{A_S\}$ is $k$-dimensional.
Now suppose an alternate factorization $\mathbf{y}_{i} = B \mathbf{b}_i$, with $\mathbf{b}_i$ $k$-sparse.  Since there are $k {m \choose k}$ vectors $\mathbf{y}$ for each support $S$, the ``pigeon-hole principle" implies that there are at least $k$ vectors $\mathbf{b}$ in the new factorization  that have the same support $S'$.  Hence, Span$\{A_S\} \subseteq \text{Span}\{B_{S'}\}$; and since $\dim(\text{Span}\{ B_{S'} \}) \leq k$, we have Span$\{A_S\} = \text{Span}\{B_{S'}\}$.  Applying Lemma \ref{subspaceLem_}, we obtain that $B$ differs from $A$ by a permutation and invertible scaling.
\end{IEEEproof}

We close this section by using a special case of Theorem \ref{unique_thm1} to solve a problem in recreational mathematics \cite{hillar2010}.

\begin{corollary}
Fix positive integers $k < n$.  Those $A \in \mathbb R^{n \times n}$ that satisfy (\ref{CSassump}) and have the property that $A \mathbf{a}$ is $k$-sparse for all $k$-sparse $\mathbf{a}$ are the matrices $PD$, where $P$ and $D$ run over permutation and invertible diagonal matrices, respectively.  
\end{corollary}
\begin{IEEEproof}
Let $\mathbf{a}_1,\ldots,\mathbf{a}_N$ be the $k$-sparse vectors from Theorem \ref{unique_thm1}.  Let $A$ be any matrix with the stated property and define $k$-sparse vectors $\mathbf{b}_i = A\mathbf{a}_i$.  Then, with $B$ as the identity matrix, the set $\{\mathbf{b}_1, \ldots, \mathbf{b}_N\}$ is another $k$-sparse coding of $A\mathbf{a}_i$.  Theorem \ref{unique_thm1} now implies that $A = PD$ for some permutation matrix $P$ and invertible diagonal matrix $D$.
\end{IEEEproof}

\section{Statements of probabilistic theorems}\label{proof_prob_Section}

We next give precise statements of our probabilistic versions of Theorem \ref{unique_thm1}, all of which rely on the following construction. 
\begin{definition}[Random drawing of sparse vectors]\label{random_def}
Given the support set for its $k$ nonzero entries, a \textit{random draw} of $\mathbf{a}$ is the $k$-sparse vector with support entries chosen uniformly from the interval $[0,1] \subset \mathbb R$, independently.  When a support set is not specified, a \textit{random draw} is a choice of one support set uniformly from all $T = {m \choose k}$ of them and then a random draw.
\end{definition}
 
\begin{theorem}\label{rand_thm1}
Fix   $n$, $k < m$, and a generation matrix $A \in \mathbb R^{n \times m}$ satisfying spark condition (\ref{CSassump}).  If $(k+1)$ $k$-sparse $\mathbf{a}_i$ are randomly drawn from each support set, then $Y = \{A\mathbf{a}_1,\ldots,A\mathbf{a}_{N}\}$ has a unique sparse coding with probability one.
\end{theorem}

The following is a direct application of this result, partially answering one of our questions from the introduction.

\begin{corollary}\label{feas_cor}
Suppose  $m$, $n$, and $k$ satisfy inequality (\ref{recovcond}).  With probability one, a random\footnote{Many natural ensembles of random matrices work, e.g., \cite[Section 4]{Baraniuk2008}.} $n \times m$ generation matrix $A$ satisfies (\ref{CSassump}).  Fixing such an $A$, with probability one a dataset $Y = \{A\mathbf{a}_1, \dots , A\mathbf{a}_{N}\}$ generated from $N = (k+1){m \choose k}$ $k$-sparse samples $\mathbf{a}_i$ uniquely prescribes $A$ and these sparse vectors $\mathbf{a}_i$ up to a fixed permutation and scaling.
\end{corollary}

We now state our third theorem.  Note that an \textit{algebraic set} is a solution to a finite set of polynomial equations.  

\begin{theorem}\label{rand_thm2}
Fix positive integers $n$ and $k < m$.  If $(k+1)$ $k$-sparse $\mathbf{a}_i$ are randomly drawn from each support set, then with probability one the following holds.  There is an algebraic set $Z \subset \mathbb R^{n \times m}$ of Lebesgue measure zero with the following property:  if $A \notin Z$, then $Y = \{A\mathbf{a}_1,\ldots,A\mathbf{a}_{N}\}$ has a unique sparse coding.\footnote{The spark condition itself defines an algebraic set of matrices.  Thus, the measure zero set $Z$ here can be defined without mention of (\ref{CSassump}).}
\end{theorem}

\begin{corollary}\label{feas_cor2}
Suppose  $m$, $n$, and $k$ obey inequality (\ref{recovcond}).  With probability one, a random draw of $N = (k+1){m \choose k}$ $k$-sparse samples $\mathbf{a}_i$ satisfies:  almost every matrix $A$ gives $Y = \{A\mathbf{a}_1,\ldots,A\mathbf{a}_{N}\}$ a unique sparse coding.
\end{corollary}

It is somewhat surprising that Theorems \ref{rand_thm1} and \ref{rand_thm2} do not automatically prove Theorem \ref{unique_thm1}.  We illustrate some of the subtlety between these three results as follows.   Set\[g(A,a) = A-a, \ \  A \in [0,1], \ a \in [0,1],\] and consider the following three statements. (i) There is an $a$ satisfying: for every $A$, we have $g(A,a) \neq 0$.  (ii) Fix $A$.  With probability one, a random $a$ will satisfy $g(A,a) \neq 0$.  (iii) With probability one, a random $a$ will satisfy: for almost every $A$ (i.e., outside of a set of Lebesgue measure zero), we have $g(A,a) \neq 0$.  While (ii) and (iii) are easily seen to be true, statement (i) is false (for this particular $g$).   Given this possible technicality, it is interesting that a deterministic statement of the form (i) can be made for the general uniqueness problem we study in this paper.

Experimental verification that SDL robustly satisfies both implications (\ref{PDaeqB}) of the above theorems appears in \cite{aharon2006, NIPS10}.  In particular, even if samplings (\ref{generated_data}) are not exact but contain measurement inaccuracy, SDL is still observed to succeed.  These findings suggest that noisy versions of the above results and \cite[Theorem 3]{aharon2006} hold, a focus of future work.

\section{Proofs of probabilistic theorems}\label{prob_proofs}

In this section, we prove Theorems \ref{rand_thm1} and \ref{rand_thm2}.   Our general perspective is algebraic.  We consider the $n \times m$ generation matrix $A$ as a matrix of $nm$ indeterminates $A_{ij}$ ($i=1,\ldots,n$; $j = 1,\ldots,m$).  Such matrices become actual real-valued matrices when real numbers are substituted for all the indeterminates.  For each support set $S = \{S^1,\ldots,S^k\} \in {[m] \choose k}$ with $S^1 < \cdots < S^k$, we also consider the following $k$-sparse vectors of indeterminates:
\begin{equation}\label{a_samplings}
\mathbf{a}_{S,\ell} = t^1_{S,\ell} \mathbf{e}_{S^1} + \cdots  + t^k_{S,\ell} \mathbf{e}_{S^k}, \ \  \ell = 1,\ldots, k+1.
\end{equation}
In a moment, each indeterminate $t^j_{S,\ell}$ will represent an i.i.d. draw from the  uniform distribution on the interval $[0,1]$.  

Our main object of interest is the following dataset of $N = (k+1){m \choose k}$ subsampled vectors of indeterminates:
\begin{equation}\label{Y_indeterm}
Y = \left\{ A\mathbf{a}_{S,\ell} : \ S \in {[m] \choose k}, \ \ell = 1,\ldots,k+1\right\}.
\end{equation}
We would like to show that under appropriate substitutions for the indeterminates, the dataset $Y$ has a unique sparse coding.  This construction follows the intuition in \cite[Theorem 3]{aharon2006}, which is to cover every subspace spanned by $k$ columns of $A$ with a sufficient number of generic vectors. 

Consider any $k$ samples $\{A\mathbf{a}_{S_1,\ell_1}, ,\ldots, A\mathbf{a}_{S_k,\ell_k}\}$, and let $M$ be the $n \times k$ matrix formed by stacking them horizontally:
\[ M = (A\mathbf{a}_{S_1,\ell_1}  \cdots  A\mathbf{a}_{S_k,\ell_k}).\]
Let $g(A,\mathbf{t})$ be the polynomial defined as the sum of the squares of all $k \times k$ sub-determinants of $M$.  This polynomial involves the  indeterminates $t^j_{S_1,\ell_1},\ldots, t^j_{S_k,\ell_k}$ ($j = 1,\ldots,k$), as well as indeterminates contained in columns $A_{S_1},\ldots,A_{S_k}$.  By construction, the polynomial $g$ evaluates to a nonzero number for a substitution of real numbers for the indeterminates if and only if the  $k$ columns of $M$ span a $k$-dimensional space.\footnote{If column rank of $M$ is $k$, then (since row rank equals column rank) there are $k$ linearly independent rows;  the determinant of this submatrix is nonzero.}

Fix a real-valued matrix $A$ satisfying spark condition (\ref{CSassump}).  We first claim that for almost every substitution for the indeterminates $t^j_{S,\ell}$ with real numbers, the polynomial $g$ evaluates to a nonzero real number.  Note  this would imply that with probability one all subsets of $Y$ with $k$ elements are linearly independent (as a finite union of sets with measure zero also has measure zero).  Somewhat surprisingly, to verify this claim about $g$, it is enough to show that \textit{some} substitution of real numbers for the indeterminates gives a nonzero real value for $g$.  Perhaps because we feel this fact should be more well-known, we state it here.  It can be proved by induction on the number of indeterminates and Fubini's theorem in real analysis (the base case being that a nonzero univariate polynomial has finitely many roots).  We note that the same statement also holds for real analytic $g$.

\begin{proposition}\label{poly_prop}
If a polynomial $g$ is not the zero polynomial, its zeroes are a set with Lebesgue measure zero.
\end{proposition}

\begin{lemma}\label{k_lin_ind_lem}
Fix $A \in \mathbb R^{n \times m}$ satisfying (\ref{CSassump}). With probability one, no subset of $k$ vectors from $Y$ is linearly dependent.
\end{lemma}
\begin{proof}
By Proposition \ref{poly_prop}, we need only show that each of the above $g$ can be nonzero.  However, this is clear since we can choose $t^j_{S,\ell}$ so that $M$ consists of $k$ different columns of $A$, which are linearly independent by assumption.
\end{proof}

\begin{lemma}\label{eq_supp_lem}
Fix $A \in \mathbb R^{n \times m}$ satisfying (\ref{CSassump}).  With probability one, if $k+1$ vectors $\mathbf{y}_j \in Y$ are linearly dependent, then all of the $k+1$ vectors $\mathbf{a}_j$ have the same supports.
\end{lemma}
\begin{proof}
Consider the polynomial $h$ which is the sum of the squares of all $(k+1) \times (k+1)$ minors of the $n \times (k+1)$ matrix of columns $\mathbf{y}_1, \ldots, \mathbf{y}_{k+1}$.
Suppose that the supports $S_1, \ldots, S_{k+1}$ of the $k+1$ vectors in $Y$ have more than $k$ different indices.  Then since any $k+1$ columns of $A$ are linearly independent, there is a setting for the $t^j_{S,\ell}$ such that $h$ is a nonzero real number.   As before, it follows that for almost every substitution of the $t^j_{S,\ell}$ with real numbers, the polynomial $h$ evaluates to a nonzero number.  In particular, with probability one, if the chosen $\mathbf{y}$ are linearly dependent, then $S_1 = \cdots = S_{k+1}$.
\end{proof}

\begin{proof}[Proof of Theorem \ref{rand_thm1}]
Fix a real matrix $A$ satisfying (\ref{CSassump}) and consider any alternate factorization $\mathbf{y}_i = B\mathbf{b}_i$ for $Y$.  Given a support set $S \in {[m] \choose k}$, let $J(S) = \{ j : \text{support}(\mathbf{b}_j) \subseteq S\}$, and note that those $\mathbf{y}$ indexed by $J(S)$ span at most a $k$-dimensional space.  By Lemma \ref{eq_supp_lem}, with probability one, either $|J(S)| \leq k$ or all $\mathbf{a}_j$ with $j \in J(S)$ have the same support $S'$.  Since the number of $\mathbf{a}_j$ with each support pattern $S$ is $k+1$, we have $|J(S)| \leq k+1$, and since there are a total of 
$N=(k+1){n \choose k}$ 
samples, it follows that $J(S) = k+1$ for each $S$. 
Thus, by Lemma \ref{eq_supp_lem} again, we can define a map $\phi: {[m] \choose k} \to {[m] \choose k}$ that takes each subset $S$ to the unique $S'$ which is the support of $\mathbf{a}_j$ for each $j \in J(S)$.  One easily checks that $\phi$ is injective and thus bijective, and also that (\ref{eq:spanhypoth}) holds where $\pi = \phi^{-1}$.  Thus, by Lemma \ref{subspaceLem_} there is a permutation $P$ and invertible diagonal $D$ such that $A = BPD$.   
\end{proof}

In the proof above, we first picked a matrix $A$ and then verified that a random setting of the $t^j_{S,\ell}$ determine a dataset $Y$ with a unique sparse coding.  Ideally, our sparse vectors that sample to $Y$ should be independent of the CS matrix $A$, as in Theorem \ref{unique_thm1}.  The closest we come to this using random methods is Theorem \ref{rand_thm2} from Section \ref{proof_prob_Section}, which we now prove.

\begin{proof}[Proof of Theorem \ref{rand_thm2}]
Consider a polynomial $g$ before the proof of Lemma \ref{k_lin_ind_lem} or a polynomial $h$ as in the proof of Lemma \ref{eq_supp_lem}.  By Proposition \ref{poly_prop}, we have for almost every setting of the $t^j_{S,\ell}$ that $g$ (resp. $h$) is not the zero polynomial in the indeterminates $A_{ij}$.  In particular, for almost every $A$, it is not the real number zero.  The proof remainder is as above.
\end{proof}

Our probabilistic theorems require $k+1$ points for each of the $T = {m \choose k}$ sparse supports.  How many samples $N$ are needed if samples are not systematically generated to meet this requirement but simply drawn randomly from a large set data?  If we require the failure probability to be bounded above by some small $\beta$, then as we explain $N = \frac{k+1}{\beta}{m \choose k}$ is sufficient. 

Consider first the number of ways for each of the $T$ support sets to contain at least $k+1$ points of the $N$ samples: it is the number of nonnegative integer solutions $d_i$ to the equation \[ (d_1 + k+1) + \cdots + (d_T + k +1) = N.\]
More generally, the number of solutions to such an equation $d_1 + \cdots + d_T = M$ is well-known to number ${M + T - 1 \choose T - 1}$.
It follows that $N$ samples (with supports chosen uniformly) fail to have $k+1$ points in each of $T$ support sets with a probability 
\[q = \frac{{N + T - 1 \choose T -1} - {N - Tk - 1 \choose T -1 }}{{N + T - 1 \choose T -1}} = 1 - 1/r.\]

Our goal therefore reduces to finding a sufficient quantity of samples $N = \alpha T$ guaranteeing a failure probability at most $\beta \geq q$.  
To achieve this, we shall show that for large $\alpha$, the quantity  $r$ is asymptotic to $1 + (k+1)/\alpha$; in particular, $q \approx (k+1)/\alpha$ and to obtain $q \leq \beta$, we need $N \geq \frac{k+1}{\beta}T$.
First, we approximate $r$ as follows:
\begin{eqnarray}
r &=& \prod_{i=1}^{T-1} \frac{(N+T-i)}{(N-Tk-i)} = \prod_{i=1}^{T-1} \left( 1 + \frac{k+1}{(N/T-k-i/T)}\right)\nonumber\\ 
  &=& 1 + (k+1) \sum_{i=1}^{T-1} \frac{1}{\alpha - k-i/T} + O(\alpha^{-2}).\nonumber
\end{eqnarray}
And then we estimate the sum on the right using calculus: 
\begin{eqnarray}
\sum_{i=1}^{T-1} \frac{1}{\alpha - k-i/T}  &\approx&  \int_{1/T}^{1} \frac{dx}{\alpha - k-x} = \ln \frac{\alpha - k-1/T}{\alpha - k-1}\nonumber\\ 
&\approx& \ln \frac{\alpha - k}{\alpha - k-1} \approx \frac{1}{\alpha}.\nonumber
\end{eqnarray}

We close this section by sketching proofs for the corollaries stated in Section \ref{proof_prob_Section}.

\begin{proof}[Proofs of Corollaries \ref{feas_cor},  \ref{feas_cor2}]
From CS theory (e.g., \cite{Baraniuk2008}), there exist matrices with the restricted isometry property (RIP) for any $m$, $n$, and $k$ satisfying (\ref{recovcond}).  Thus, using the ideas presented here, one can construct a nonzero polynomial in the entries of a general matrix $A$ that vanishes if and only if $A$ fails to obey the spark condition. In particular, almost every matrix $A$ satisfies (\ref{CSassump}).  The result now follows directly from Theorem \ref{rand_thm1};  Corollary \ref{feas_cor2} is proved similarly.
\end{proof}

\section{Discussion}\label{discussion}

In this article, we have proved theorems specifying bounds on the number of samples $N$ sufficient for guaranteeing uniqueness in sparse dictionary learning.  Such bounds depend critically on the precise problem statement.  In particular, a guarantee with certainty that {\it any} generation matrix $A$ is recoverable  (Theorem \ref{unique_thm1}) requires more samples than a guarantee with probability one for a particular $A$ (Theorem \ref{rand_thm1}) or for all but a Lebesgue measure zero set (Theorem \ref{rand_thm2}). We note that the lower bound for achieving the deterministic guarantee given here is a drastic tightening of an earlier bound \cite{hillar2011ramsey}, which used Ramsey theory. 

We emphasize that our theorems rely on systematic construction of training samples, which for Theorems \ref{rand_thm1} and \ref{rand_thm2} include random drawing. It is an open problem how much our bounds can be strengthened.  In regimes of moderate compression, computer experiments (e.g., \cite{aharon2006, NIPS10}) suggest that the amount of data needed for successful SDL is smaller than our estimates.

Another important item left unaddressed here is to find conditions under which sparse dictionary learning is guaranteed to converge. Although widely used in practice, sparse dictionary learning algorithms are typically non-convex, and finding a proof of convergence for  an SDL scheme can be challenging \cite{mairal2010online, gribonval2010, balavoine2012convergence, spielman2012, arora2013new, agarwal2013exact, agarwal2013learning, arora2014more}. 

{\it Comments about the mathematics:}  A mathematical technique initiated by Szele and Erd\H{o}s \cite{szele1943, erdos1947} called ``the probabilistic method" \cite{alon2011probabilistic} produces combinatorial structures using randomness that are difficult to find deterministically, much as we have done here.  We note that the general problem of finding deterministic constructions of objects easily formed using randomness can be very difficult.  For instance, how to deterministically construct optimal compressed sensing RIP matrices is still open, with the best results so far being \cite{bourgain2011,li2012} (see also \cite{monajemi2013} for a recent, large experimental study).  

{\it Consequences for applications:} A practical consequence of our work is the description of a feasible regime of overcompleteness for universality in sparse dictionary learning (Corollaries \ref{feas_cor} and \ref{feas_cor2}). Another interesting implication is that obvious structure in a learned dictionary should not be the only criterion of success for SDL. For instance, if compression by a matrix $\Phi$ is involved, then columns of a learned dictionary $B = \Phi \Psi D^{-1} P^{\top}$ might not reveal visually salient structure even though learning has converged and the resulting sparse codes accurately represent the underlying sparse structure of original data (see Fig.~\ref{ACS_fig_random} for an example of this phenomenon). Even though the dictionary might be difficult to interpret in the case of compression, extraction of sparse components might still be useful for subsequent steps of analysis.  For example, recent work has demonstrated that sparse coding of sensory signals can significantly improve performance in classification tasks \cite{yang2009,boureau2010,coates2011,rigamonti2011}.   It is possible that compressing data with a random matrix in SDL can reduce the number of model parameters, thereby speeding training.

{\it Implications for neuroscience:} Our work was originally motivated by the problem of communication between brain regions in theoretical neuroscience \cite{NIPS10}. Neural representations of sensory input can employ large numbers of local neurons in a brain area. However, the local neurons with axonal fibers projecting into a second brain region usually comprise a small fraction of all local neurons.  Thus, there is potentially an information bottleneck in the communication between brain areas. Assume  that the firing patterns in the sender region are large sparse vectors $\mathbf{a}$ and that the transmission to the receiver brain area can be described by compression with a linear transformation $\Phi$. Decoding of the compressed signals $\mathbf{y}$ in the receiving brain area can now be performed by overcomplete SDL, implemented by synaptic learning in local neural networks.  As shown in this paper, a decoding network that successfully reconstructs enough compressed data necessarily recovers original sparse representations. Such a theory of brain communication is compatible with early theories of efficient coding in the brain \cite{Attneave1954, barlow1959,olshausenfield1996} with one important difference.  The learning objective of traditional efficient coding is the faithful reconstruction of a sensory signal. In the adjusted theory \cite{NIPS10}, the objective is reconstruction of the locally available compressed signals. Our results here describe conditions under which learning driven by this weaker objective can still uniquely recover the sparse components of the full sensory signal. Of course, without compression, the learning objective coincides with that of traditional efficient coding.  For some other examples of compressed sensing being applied to neuroscience, see the survey \cite{ganguli2012}.

Our final application of uniqueness in SDL is to the analysis of neuroscience data. With measurement techniques improving, the analysis of signals from dense electrode arrays placed in the brain has become a pressing issue.  In particular, if stimulus features are not topographically organized in a neuronal circuit, each electrode might sense a mixture of signals from many neurons tuned to different features. Such a situation occurs in multi-electrode recordings from the hippocampal brain region in navigating rats. Individual neurons in this region encode different locations of the animal in its environment. Because neighboring cells do not encode similar locations, the local field potentials have no pronounced place selectivity and have been regarded as largely uninformative about the animal's location. Recently, however, performing SDL on these vastly subsampled data has led to the discovery of sparse place-selective field components. The extracted components tile the entire environment of the animal and encode its precise instantaneous localization \cite{agarwal2014spatially}. Uniqueness guarantees of the type presented here are important to ensure that the extracted behaviorally relevant components reflect structure of the data and are independent of the details of the particular SDL algorithm used for data analysis.

\appendix[I. Combinatorial matrix theory]\label{comb_mat_lem}
In this section, we prove Lemma \ref{subspaceLem_}, which was a main ingredient in proofs of our main theorems. We suspect that there is an appropriate generalization to matroids.  
First, however, we state the following easily deduced facts.
\begin{lemma}\label{basicSpan_lem}
Let $M \in \mathbb R^{n \times m}$.  If every set of $\ell + 1$ columns of $M$ are linearly independent, then for $S,S' \in  {[m] \choose \ell}$,
\begin{equation}\label{kplus1Fact}
\text{\rm Span}\{M_S\} = \text{\rm Span}\{M_{S'}\}  \  \Longrightarrow \  S = S'.
\end{equation}
If $M$ satisfies condition (\ref{CSassump}) and $S_1,S_2 \in  {[m] \choose k}$, then
\begin{equation}\label{2kIndFact}
 \text{\rm Span}\{M_{S_1 \, \cap \, S_2}\} =   \text{\rm Span}\{M_{S_1}\} \cap  \text{\rm Span}\{M_{S_2}\}.
 \end{equation}
\end{lemma}
%

\begin{IEEEproof}[Proof of Lemma \ref{subspaceLem_}]
We shall induct on $k$; the base case $k = 1$ is worked out at the beginning of Section \ref{acsTheoremSection}. 
We first prove that $\pi$ is injective (and thus bijective).  Suppose that $S_1, S_2 \in {[m] \choose k}$ have $\pi(S_1) = \pi(S_2)$; then by (\ref{eq:spanhypoth}), 
\begin{equation*}
\text{\rm Span}\{A_{S_1}\} = \text{\rm Span}\{B_{\pi(S_1)}\} =  \text{\rm Span}\{B_{\pi(S_2)}\} =  \text{\rm Span}\{A_{S_2}\}.
\end{equation*}
In particular, using (\ref{kplus1Fact}) from Lemma \ref{basicSpan_lem} above with $\ell = k$ and $M = A$, it follows that $S_1 = S_2$ and thus $\pi$ is bijective.  Moreover, from this bijectivity of $\pi$ and the fact that every $k$ columns of $A$ are linearly independent, it follows that every $k$ columns of $B$ are also linearly independent.

We complete the proof, inductively, by producing a map:
\begin{equation}\label{inductive_tau}
\tau:  {[m] \choose k-1} \to {[m] \choose k-1}
\end{equation}
which satisfies $\text{\rm Span}\{A_{S}\} = \text{\rm Span}\{B_{\tau(S)}\} \ \text{for} \ S \in {[m] \choose k-1}$.
Let $\alpha = \pi^{-1}$ denote the inverse of $\pi$.  Fix $S = \{i_1,\ldots,i_{k-1}\} \in  {[m] \choose k-1}$, and set $S_1 = S \cup \{r\}$ and $S_2 = S \cup \{s\}$ for some $r,s \notin S$ with $r \neq s$ (so that $\alpha(S_1) \neq \alpha(S_2)$ by injectivity of $\alpha$).\footnote{Here we use the assumption that $k < m$ so that such a pair $r \neq s$ exists.}  Intersecting equations (\ref{eq:spanhypoth}) with $S = \alpha(S_1)$ and $S = \alpha(S_2)$  and then applying identity (\ref{2kIndFact}) with $M = A$, it follows that 
\begin{equation}\label{BsAteqn1}
 \text{\rm Span}\{B_{S},B_{r}\} \cap \text{\rm Span}\{B_{S},B_{s}\} =\text{\rm Span}\{A_{\alpha(S_1) \, \cap \, \alpha(S_2)}\}.
 \end{equation}
Since the left-hand side of (\ref{BsAteqn1}) is at least $k-1$ dimensional, the number of elements in the set $\alpha(S_1)  \cap \alpha(S_2)$ is either $k-1$ or $k$.  But $\alpha(S_1) \neq \alpha(S_2)$ so that $\alpha(S_1)  \cap \alpha(S_2)$ consists of $k-1$ elements.
Moreover,  $\text{\rm Span}\{B_{S}\} \subseteq \text{\rm Span}\{A_{\alpha(S_1) \,  \cap \, \alpha(S_2)}\}$
implies that $\text{\rm Span}\{B_{S}\} = \text{\rm Span}\{A_{\alpha(S_1) \,  \cap \, \alpha(S_2)}\}$.

The association $S \mapsto \alpha(S_1)  \cap \alpha(S_2)$ above defines a function $\gamma: {[m] \choose k-1} \to {[m] \choose k-1}$ with $\text{\rm Span}\{B_{S}\} = \text{\rm Span}\{A_{\gamma(S)}\}$.   Finally, we show that $\gamma$ is injective, which implies that $\tau = \gamma^{-1}$ is the map desired in (\ref{inductive_tau}) for the induction.
If $\gamma(S) = \gamma(S')$, then $\text{\rm Span}\{B_{S}\} = \text{\rm Span}\{B_{S'}\}$.  
By (\ref{kplus1Fact}) in Lemma \ref{basicSpan_lem} with $\ell = k-1$ and $M = B$, we have $S = S'$.  Thus, $\gamma$ is injective.
\end{IEEEproof}

\begin{example}\label{prop_example}
We show how the proof of Lemma \ref{subspaceLem_} works in the case $n = m = 3$, $k = 2$.  Suppose that $\pi: {[3] \choose 2} \to {[3] \choose 2}$ is \[\pi(\{1,2\}) = \{2,3\}, \, \pi(\{1,3\}) = \{1,2\}, \,\pi(\{2,3\}) = \{1,3\}.\] 
Following the proof of Lemma \ref{subspaceLem_}, one can check that 
\[ \gamma(\{1\}) = \{3\}, \ \gamma(\{2\}) = \{1\},\  \gamma(\{3\}) = \{2\}, \]
and thus we obtain the map $\tau = \gamma^{-1}$ as desired in (\ref{inductive_tau}).  The resulting permutation $P$ is the cycle $1 \mapsto 2 \mapsto 3 \mapsto 1$.
\end{example}

\section*{Acknowledgment}
We thank the following people for helpful discussions: Charles Cadieu, Melody Chan, Will Coulter, Jack Culpepper, Mike DeWeese, Rina Foygel, Guy Isely, Amir Khosrowshahi, Matthias Mnich, and Chris Rozell.  We also thank the anonymous referees for comments that considerably improved this work, including supplying ideas for proofs of several results.  This work was supported by grant IIS-1219212 from the National Science Foundation. CJH was also partially supported by an NSF All-Institutes Postdoctoral Fellowship administered by the Mathematical Sciences Research Institute through its core grant DMS-0441170.  FTS was also supported by the Applied Mathematics Program within the Office of Science Advanced Scientific Computing Research of the U.S. Department of Energy under contract No. DE-AC02-05CH11231.

\ifCLASSOPTIONcaptionsoff
  \newpage
\fi



\bibliographystyle{IEEEtran}
\bibliography{acs}
%

%
%

%


\begin{IEEEbiographynophoto}{Christopher J. Hillar}
completed a B.S. in Mathematics and a B.S. in Computer Science at Yale University.  Supported by an NSF Graduate Research Fellowship, he received his Ph.D. in Mathematics from the University of California, Berkeley in 2005. From 2005 - 2008, he was a Visiting Assistant Professor and NSF Postdoctoral Fellow at Texas A \& M University. From 2008 - 2010, he was an
NSF Mathematical Sciences Research Institutes Postdoctoral Fellow at  MSRI in Berkeley, CA.  In 2010, he joined the Redwood Center for Theoretical Neuroscience at UC Berkeley.
\end{IEEEbiographynophoto}
\begin{IEEEbiographynophoto}{Friedrich T. Sommer}
received the diploma in physics from the University of Tuebingen, in 1987, the Ph.D. in physics from the University of Duesseldorf, in 1993, and his habilitation in computer science from University of Ulm, in 2002. He is an Adjunct Professor at the Redwood Center for Theoretical Neuroscience and at the HelenWills Neuroscience Institute, University of California, Berkeley, since 2005. In 2003 - 2005, he was a Principal Investigator at the Redwood Neuroscience Institute, Menlo Park, CA. In 1998 - 2002, he was an Assistant Professor at the Department of Computer Science, University of Ulm, after completing Postdocs at Massachusetts Institute of Technology, Cambridge, MA, and the University of Tuebingen.
\end{IEEEbiographynophoto}






\end{document}